\documentclass[11pt,oneside]{amsart}
\usepackage{amssymb}
\usepackage{amsfonts}
\usepackage{amsmath}
\usepackage{amsthm}
\usepackage{dsfont}
\usepackage{mathtools}
\usepackage{geometry}
\usepackage[colorlinks=true,linkcolor=blue,urlcolor=blue, hypertexnames=false]{hyperref}

\mathtoolsset{showonlyrefs}
\usepackage{xcolor}
\usepackage{float}
\usepackage{comment}
\usepackage{enumerate}
\usepackage{bbm}
\usepackage{bm}
\usepackage{natbib}
\usepackage{setspace}
\newcommand{\package}[1]{\texttt{#1}}
\newtheorem{theorem}{Theorem}

\theoremstyle{definition}

\renewcommand{\mid}{\,|\,}

\numberwithin{equation}{section}

\title{Fragility Index for Time-to-Event Endpoints in Single-Arm Clinical Trials} 
\author[A.\ Maity, J.\ Garg, C.\ Basu]{Arnab Kumar Maity$^{\dagger}$, Jhanvi Garg$^{*}$, AND Cynthia Basu$^{\S}$}

\thanks{$^{\dagger}${Boehringer Ingelheim, 900 Ridgebury Rd, Ridgefield, CT 06877}, $^{*}${Department of Statistics, Texas A\&M University, College Station, TX 77843, USA}, $^{\S}${Pfizer Inc., 10555 Science Center Dr, San Diego, CA 92121}}
\thanks{$^{*}$Corresponding author (email: gargjhanvi@tamu.edu)}

\begin{document}
\begin{abstract}
The reliability of clinical trial outcomes is crucial, especially in guiding medical decisions. In this paper, we introduce the \textbf{Fragility Index (FI)} for time-to-event endpoints in single-arm clinical trials—a novel metric designed to quantify the robustness of study conclusions. The FI represents the smallest number of censored observations that, when reclassified as uncensored events, causes the posterior probability of the median survival time exceeding a specified threshold to fall below a predefined confidence level. While drug effectiveness is typically assessed by determining whether the posterior probability exceeds a specified confidence level, the FI offers a complementary measure, indicating how robust these conclusions are to potential shifts in the data. Using a Bayesian approach, we develop a practical framework for computing the FI based on the exponential survival model. To facilitate the application of our method, we developed an  \textsf{R} package \package{fi}, which provides a tool to compute the Fragility Index.  Through real world case studies involving time to event data from single arms clinical trials, we demonstrate the utility of this index. Our findings highlight how the FI can be a valuable tool for assessing the robustness of survival analyses in single-arm studies, aiding researchers and clinicians in making more informed decisions. 

\vspace{0.3cm}

\noindent\textbf{Keywords:} {Bayesian analysis, Exponential survival model, Fragility Index, Single-arm clinical trials, Time-to-event data}
\end{abstract}

\maketitle
\section{Introduction}\label{sec:intro}
The Fragility Index (FI) is a widely used metric in randomized controlled trials (RCTs) to assess the robustness of statistically significant findings, particularly in studies with small sample sizes or limited outcome events. Traditional significance testing, which relies on p-values, can sometimes create a false sense of confidence, as minor changes in the data may shift results from significant to non-significant, casting doubt on the reliability of conclusions \citep{editorial2019s}. For binary outcomes, the FI, introduced by \citet{walsh2014statistical}, quantifies the vulnerability of a study's results by determining the minimum number of events (e.g., successes or failures) that would need to change to reverse statistical significance. This is calculated by systematically altering the outcomes—changing non-events to events or vice versa—until the p-value exceeds the significance threshold, typically $p > 0.05$. For example, in a trial testing a new treatment, if statistical significance is observed, a low FI, meaning that just one or two more events would change the result, indicates that the findings are fragile and may lack robustness. Although there is no universal threshold for robustness, a higher FI generally indicates more reliable results \citep{andrade2020use}. To adjust for sample size, the Fragility Quotient (FQ), calculated as FI/sample size, has been proposed to provide a normalized measure of fragility \citep{heston2023robustness}. Understanding fragility allows researchers and clinicians to interpret trial results with more caution, particularly when clinical decisions are based on fragile evidence \citep{garcia2023fragility}.

Several studies have highlighted the practical implications of the FI. \citet{meyer2014fibrinolysis} demonstrated that in a pulmonary embolism trial, altering the outcomes of only three patients rendered the results non-significant, exposing the limited robustness of the findings. Similarly, \citet{tignanelli2019fragility} emphasized the value of the FI in identifying fragile findings and improving patient care by fostering a deeper understanding of trial robustness. In critical care research, \citet{ridgeon2016fragility} revealed that many trials reporting statistically significant effects on mortality had low Fragility Index scores, suggesting that their results relied heavily on a small number of events, raising concerns about their robustness. With the development of tools like the \package{fragility} \textsf{R} package \citep{lin2022assessing}, calculating and visualizing the FI for clinical trials with binary outcomes has become more accessible. However, the FI has some limitations, particularly its reliance on binary outcomes, small sample sizes, and the use of Fisher’s exact test, which can affect its accurate calculation, especially in studies involving time-to-event data \citep{baer2021fragility, potter2020dismantling}.

To overcome these limitations in survival analysis, the concept of fragility has been extended to time-to-event data in two arm trials with the introduction of the Survival-Inferred Fragility Index (SIFI) \citep{bomze2020survival}. The SIFI measures the robustness of statistically significant findings in clinical trials with survival endpoints by identifying the minimum number of patients with the longest survival times in treatment group who must be reclassified to the control group to reverse statistical significance, typically assessed using a two-sided log-rank test ($P > 0.05$). This provides a more precise measure of how sensitive survival outcomes are to changes in the data, offering deeper insights into trial stability. \citet{liu2024fragility} assessed the robustness of 332 phase III oncology trials using the SIFI and found that trials involving targeted therapies, progression-free survival endpoints, and positive outcomes tend to be the most robust. Similarly, \citet{olsen2024statistical} analyzed 113 pediatric oncology phase 3 trials and found that pediatric trials were similarly or less fragile compared to adult oncology trials.

Historically, the FI has been used primarily in two-arm or multi-arm trials. In this paper, we extend the concept of the Fragility Index to time-to-event data in single-arm trials and explore its practical applications. Section \ref{sec:method} provides a detailed definition of the FI for time-to-event data in single arm trials, explaining how it measures the robustness of trial outcomes. Section \ref{sec:simulation} presents case studies using real-world survival datasets, illustrating the usefulness of the FI in practice. Finally, Section \ref{sec:conclusion} discusses the broader implications of using the FI for survival analysis, emphasizing its value as a supplementary metric for interpreting trial results and aiding clinical decision-making.

\section{Methodology}\label{sec:method}

In clinical trials, particularly in fields like oncology and chronic diseases, outcomes are frequently associated with the time until a significant event occurs, such as disease progression or death. This leads to the analysis of time-to-event data, also known as survival data. Survival analysis not only captures whether the event occurs but also when it happens, offering richer insights into the treatment's efficacy and impact. These time-based outcomes are critical in understanding treatment effects over the course of a study and are a common focus in clinical research involving chronic or progressive conditions. A common challenge in survival analysis is the presence of censoring—a scenario where the exact time of the event is not fully observed. Right-censoring, the most common type, occurs when a patient either has not experienced the event by the end of the study or is lost to follow-up before the event occurs.
 In these cases, the patient’s data is incomplete, as we only know that the event did not occur up to a certain point in time, but the exact time (if or when it will occur) remains unknown.  Analyzing such data poses complexities, particularly in accounting for censored observations.

To model time-to-event data in survival analysis, one of the widely used distributions is the exponential distribution. Its popularity can be traced back to the seminal works of Epstein \citep{epstein1953life, epstein1954some, epstein1954truncated}, who demonstrated its practicality and versatility in modeling life data and survival times. The exponential distribution has the \textit{memoryless} property, meaning the hazard rate remains constant over time. This implies that the likelihood of the event occurring at any moment is independent of how much time has already passed, making it a suitable choice in situations where the event risk is expected to stay stable throughout the study.

The probability density function (pdf) of the exponential distribution is:
\[
f(t) = \lambda e^{-\lambda t}, \quad t \geq 0
\]
where \(\lambda\) is the rate parameter. The expected time to the event is the inverse of the rate parameter, \(1/\lambda\), providing an estimate of the average time until the event occurs. The survival function, which gives the probability that a patient survives beyond a given time \(t\), is:
\[
S(t) = e^{-\lambda t}
\]
and the hazard function, which describes the instantaneous rate of the event occurrence given that the patient has survived up to time \(t\), is constant and given by:
\[
h(t) = \lambda
\]

\subsection{Likelihood Function with Censoring}
In the presence of censoring, the likelihood function must account for both observed events and censored data. Suppose there be $n$ independent observations and for $i \leq i \leq n$, let \((T_i, \delta_i)\) represent the observed time and censoring indicator for patient \(i\), where \(\delta_i = 1\) if the event is observed and \(\delta_i = 0\) if the observation is censored. The likelihood function for \(n\) independent observations is given by:
\begin{equation}
L(\lambda) = \prod_{i=1}^n \left[ \lambda e^{-\lambda T_i} \right]^{\delta_i} \left[ e^{-\lambda T_i} \right]^{1-\delta_i}\label{eq: likelihood}
\end{equation}

\subsection{Choice of Prior}
In Bayesian analysis, the choice of prior distribution reflects initial beliefs about the parameter before observing the data. This could be based on prior studies, expert knowledge, or other relevant information available before the trial. A natural choice for the rate parameter \(\lambda\) is the Gamma prior, as it is conjugate to the exponential likelihood, simplifying posterior inference. The Gamma distribution is defined by the shape parameter \(\alpha\) and the rate parameter \(\beta\), with the probability density function:
\[
\pi(\lambda) = \frac{\beta^\alpha}{\Gamma(\alpha)} \lambda^{\alpha - 1} e^{-\beta \lambda}, \quad \lambda > 0
\]
Using Bayes’ theorem, the posterior distribution of \(\lambda\) is derived by combining the prior and the likelihood from the observed data. The posterior distribution remains a Gamma distribution with updated parameters:
\[
\alpha' = \alpha + \sum_{i=1}^n \delta_i, \quad \beta' = \beta + \sum_{i=1}^n T_i
\]

\begin{theorem}[Posterior Distribution of Rate Parameter $\lambda$]\label{thm:posterior_censoring}
Given \(n\) independent observations \((T_i, \delta_i)\) for \(i=1, \dots, n\), where \(T_i\) represents the observed time and \(\delta_i\) is the censoring indicator, and assuming a Gamma prior \( \text{Gamma}(\alpha, \beta) \) for \(\lambda\), the posterior distribution of \(\lambda\) is:
\[
\lambda \mid \left(T_i, \delta_i\right)_{i=1}^{n} \sim \text{Gamma}\left( \alpha + \sum_{i=1}^n \delta_i, \, \beta + \sum_{i=1}^n T_i \right)
\]
\end{theorem}

\begin{proof}
See Appendix \ref{sec:aux}
\end{proof}
This posterior distribution provides an updated estimate of \(\lambda\), incorporating both the prior belief and the observed survival data.

\subsection{Fragility Index}
In time-to-event data, drug effectiveness is often assessed by determining whether the posterior probability that the median survival time exceeds a specified threshold is above a certain confidence level. The Fragility Index (FI) provides an additional measure of robustness  by quantifying how many events must change to bring this probability below that confidence level. Specifically, for survival data from single-arm trials, where the posterior probability that the median survival time exceeds a threshold \(t_0\) initially surpasses a predefined confidence level \(p_0\), the \textbf{Fragility Index (FI)} is defined as the smallest number \(k\) of censored observations with the shortest censoring times that, when reclassified as uncensored events, reduce the posterior probability below the specified confidence level \(p_0\).

This index provides a precise measure of the trial's robustness, indicating how sensitive the results are to changes in the data. In essence, it indicates the minimum number of censored observations that must be reclassified to uncensored events to reduce confidence in the treatment effect. Mathematically, the median survival time \(t_{\text{med}}\) for an exponential distribution is given by:
\[
t_{\text{med}} = \frac{\ln 2}{\lambda}
\]

\begin{theorem}[Median Survival Time for Exponential Distribution]\label{thm: 3}
For a random variable \(T\) following an exponential distribution with rate parameter \(\lambda\), the median survival time is:
\[
t_{\text{med}} = \frac{\ln 2}{\lambda}
\]
\end{theorem}
\begin{proof}
See Appendix \ref{sec:aux}
\end{proof}

The posterior probability that the median survival time exceeds a threshold \(t_0\) is computed by integrating over the posterior distribution of \(\lambda\).

\begin{theorem}[Posterior Probability of Median Survival Time]\label{thm: main}
Given the posterior distribution of the rate parameter \(\lambda\) for an exponential survival model, with data \(\left(T_i, \delta_i\right)_{i=1}^{n}\), as $\text{Gamma}\left( \alpha + \sum_{i=1}^n \delta_i, \, \beta + \sum_{i=1}^n T_i \right)$, the posterior probability that the median survival time \(t_{\text{med}}\) exceeds a threshold \(t_0\) is:
\[
P\left( t_{\text{med}} > t_0 \mid \left(T_i, \delta_i\right)_{i=1}^{n} \right) = P\left( \lambda < \frac{\ln 2}{t_0} \mid \left(T_i, \delta_i\right)_{i=1}^{n} \right).
\]
This probability is computed by integrating the posterior Gamma distribution:
\[
P\left( t_{\text{med}} > t_0 \right) = \int_0^{\frac{\ln 2}{t_0}} \frac{\beta'^{\alpha'}}{\Gamma(\alpha')} \lambda^{\alpha' - 1} e^{-\beta' \lambda} \, d\lambda,
\]
where \(\alpha' = \alpha + \sum_{i=1}^n \delta_i\) is the updated shape parameter, and \(\beta' = \beta + \sum_{i=1}^n T_i\) is the updated rate parameter.
\end{theorem}

The proof of the Theorem \ref{thm: main} follows directly from Theorem \ref{thm:posterior_censoring} and \ref{thm: 3}. The Fragility Index is determined by sequentially reclassifying censored observations with the smallest censoring times to uncensored events and recalculating the posterior probability until it falls below the specified confidence level. A higher Fragility Index indicates greater robustness of the study results, implying that the treatment effect remains consistent even when multiple censoring statuses are changed. Conversely, a lower Fragility Index suggests that the results are sensitive to small changes in the data, potentially undermining confidence in the conclusions. It is important to note that there is no universally defined threshold for the FI; it serves as a relative indicator of robustness rather than an absolute measure of significance. Therefore, the FI should be used alongside other statistical measures to provide a more comprehensive evaluation of a study's findings.

\section{Numerical Study}\label{sec:simulation}

In this section, we apply the Fragility Index (FI) methodology to real-world clinical trial data to evaluate the robustness of study conclusions.  By utilizing our \href{https://github.com/arnabkrmaity/fi}{\textsf{R} package \package{fi}}\footnote{\url{https://github.com/arnabkrmaity/fi}}, we efficiently calculate the FI, demonstrating its effectiveness and ease of use in survival analysis.
\subsection{Case Study 1: Lung Cancer}

We begin with the North Central Cancer Treatment Group data, available in the \texttt{lung} dataset from the \texttt{survival} package in R. This dataset comprises observations from patients with advanced lung cancer, providing detailed information on survival times and various clinical variables. 

For this analysis, we randomly selected 30 patients from the dataset, of whom 22 had experienced the event (death), while the remaining 8 were censored. Using the observed survival times and censoring data, we calculated the posterior probability that the median survival time exceeded 7 months. A Gamma prior with shape parameter \(\alpha = 0.5\) and rate parameter \(\beta = 0.5\) was chosen for the survival rate parameter \(\lambda\), as it is weakly informative and conjugate to the exponential likelihood, making it an appropriate choice for survival analysis. The resulting posterior probability was 0.935, indicating a high likelihood that the median survival time exceeded 7 months.

The Kaplan-Meier curve for this dataset is shown in Figure \ref{fig:1}.

\begin{figure}[H]
\centering
\includegraphics[width=0.7\textwidth]{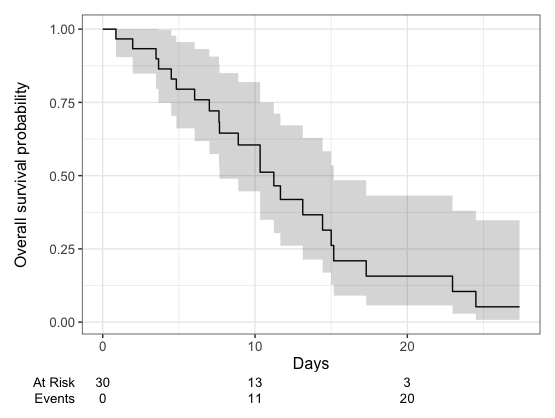}
\caption{Kaplan-Meier curve for the lung cancer dataset}
\label{fig:1}
\end{figure}

The Fragility Index was determined by sequentially reclassifying censored observations with the shortest censoring times as events and recalculating the posterior probability until it fell below the predefined confidence threshold of 0.7, a standard choice balancing statistical confidence and flexibility. For this dataset, the Fragility Index was found to be 5. This means that reclassifying five censored patients as having experienced the event (death) would reduce the posterior probability of the median survival time exceeding 7 months to below 0.7. Given a sample size of 30, an FI of 5 suggests  that the study’s conclusions are robust. This level of FI suggests that the findings of the study are stable and reliable, even with moderate changes in the data. 

\subsection{Case Study 2: Pembrolizumab in  hepatocellular carcinoma (HCC)}

For the third case study, we analyzed progression free survival time for Pembrolizumab, an immune checkpoint inhibitor widely used in the treatment of various cancers, including advanced hepatocellular carcinoma (HCC) and other malignancies. Pembrolizumab functions by blocking the programmed cell death protein 1 (PD-1) receptor, thereby enhancing the immune system's ability to detect and destroy cancer cells. The study \citep{feun2019phase} focused on its application in advanced hepatocellular carcinoma (HCC), a disease that has shown modest response rates to checkpoint inhibitors. The Individual Patient Data (IPD) from the treatment arm was extracted from the Kaplan-Meier curve presented in \citet{feun2019phase}, using the MD Anderson Cancer Center's IPD extraction tool\footnote{\url{https://biostatistics.mdanderson.org/shinyapps/IPDfromKM/}}.

The dataset, constructed using the above extraction tool, comprises 28 patients. Among them, 20 experienced disease progression, while the remaining 8 were censored. Following the same methodology as in the previous analyses, we calculated the posterior probability that the median progression-free survival time exceeded 3.5 months. A Gamma prior with shape parameter \(\alpha = 0.5\) and rate parameter \(\beta = 0.5\) was applied, consistent with our prior studies, and the cutoff for the posterior probability was set at 0.7. The resulting posterior probability was 0.958, indicating a high likelihood that the median progression-free survival time exceeded 3.5 months.

The Kaplan-Meier curve for this dataset is shown in Figure \ref{fig:3}.

\begin{figure}[H]
\centering
\includegraphics[width=0.7\textwidth]{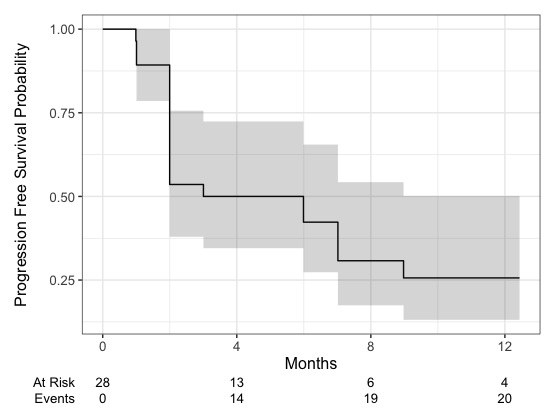}
\caption{Kaplan-Meier curve for the breast cancer dataset treated with Pembrolizumab}
\label{fig:3}
\end{figure}

The Fragility Index for this dataset was determined to be 6, meaning that reclassifying six censored patients as having experienced the event (disease progression) would reduce the posterior probability of the median progression-free survival time exceeding 3.5 months to below 0.7. Given the sample size of 28, an FI of 6 indicates that the study’s conclusions are robust. This level of FI suggests that the findings regarding Pembrolizumab’s efficacy are stable and reliable, even with moderate changes in the data. 

\subsection{Case Study 3: Palbociclib in  Breast Cancer}

For this case study, we analyzed progression free survival time from a single-arm phase II study investigating Palbociclib, a CDK4/6 inhibitor commonly used for the treatment of hormone receptor (HR)-positive, HER2-negative advanced or metastatic breast cancer (MBC). The study protocols are outlined in detail in \citet{krishnamurthy2022phase} and, as with Case Study 2, Individual Patient Data (IPD) for the treatment arm was  obtained from the Kaplan-Meier curve in \citet{krishnamurthy2022phase}.

The dataset, constructed using the MD Anderson Cancer Center's IPD extraction tool, includes 51 patients. Among them, 31 experienced disease progression, while the remaining patients were censored.

Using the same methodology as in the previous analysis, we calculated the posterior probability that the median survival time exceeded 15 months. A Gamma prior with shape parameter \(\alpha = 0.5\) and rate parameter \(\beta = 0.5\) was applied, with the cutoff for the posterior probability set at 0.7. The resulting posterior probability was 0.948, indicating a high likelihood that the median survival time exceeded 15 months.

The Kaplan-Meier curve for this dataset is displayed in Figure \ref{fig:2}.

\begin{figure}[H]
\centering
\includegraphics[width=0.7\textwidth]{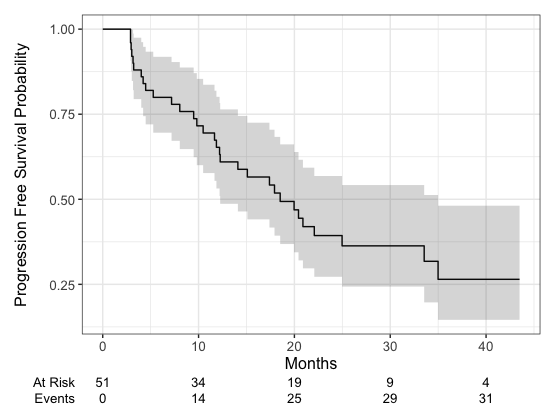}
\caption{Kaplan-Meier curve for the breast cancer dataset treated with Palbociclib}
\label{fig:2}
\end{figure}
The Fragility Index (FI) for this dataset was calculated to be 6, indicating that if six censored patients were reclassified as having experienced the event (disease progression), the posterior probability of the median survival time exceeding 15 months would drop below 0.7. With a sample size of 51, an FI of 6 suggests that the study's conclusions are moderately robust and not entirely resistant to changes in the data. This implies that the treatment effect shows some sensitivity to data alterations.

\section{Concluding Remarks}\label{sec:conclusion}

In this paper, we defined the Fragility Index (FI) for time-to-event endpoints in single-arm clinical trials, developed a Bayesian methodology using the exponential distribution, and demonstrated its application through three real-world case studies. We also  developed an \textsf{R} package \package{fi}, to facilitate the calculation of the FI, making it accessible for researchers to apply this in similar studies. The FI quantifies the robustness of study conclusions by identifying the minimum number of censored observations that, when reclassified as events, reduce the posterior probability of the median survival time exceeding a specified threshold below a confidence level. Our case studies yielded FI values of 5 and 6, indicating moderate robustness; while the findings are relatively stable, they remain somewhat sensitive to data changes. This underscores the FI's value as a complementary tool to traditional statistical methods, enhancing the interpretability and reliability of clinical trial outcomes, especially in the absence of control groups. However, the FI lacks a universal threshold and is influenced by the choice of prior distribution, necessitating careful interpretation and integration with other clinical impact measures. In conclusion, the Fragility Index represents a significant advancement in assessing the robustness of single-arm clinical trials with time-to-event data, supporting more informed and reliable clinical decision-making.

\clearpage 
\newpage 

\begin{center}
    \Large{\textbf{APPENDIX}}  
\end{center}
\medskip 
The calculation of the Fragility Index and related numerical studies are implemented in R and made accessible through the  GitHub repository \href{https://github.com/arnabkrmaity/fi}{(https://github.com/arnabkrmaity/fi)}.

\setcounter{section}{0} % Reset section counter to start with 1
\renewcommand{\thesection}{\Alph{section}} % Change numbering to letter
\renewcommand{\thetheorem}{\thesection\arabic{theorem}}
\setcounter{theorem}{0}
\renewcommand{\thelemma}{\thesection\arabic{lemma}}
\renewcommand{\thecorollary}{\thesection\arabic{corollary}}
\renewcommand{\thedefinition}{\thesection\arabic{definition}}
\renewcommand{\theexample}{\thesection\arabic{example}}
\setcounter{example}{0}
\renewcommand{\thefigure}{\thesection\arabic{figure}}
\setcounter{figure}{0}
\renewcommand{\thetable}{\thesection\arabic{table}}
\setcounter{table}{0}
\renewcommand{\theremark}{\thesection\arabic{remark}}
\setcounter{remark}{0}
\section{Proof of the Theorems} \label{sec:aux}
\begin{proof}[Proof of Theorem~\ref{thm:posterior_censoring}]
The likelihood function for the censored data is:
\[
L(\lambda) = \prod_{i=1}^n \left( \lambda e^{-\lambda T_i} \right)^{\delta_i} \left( e^{-\lambda T_i} \right)^{1 - \delta_i}
= \lambda^{\sum_{i=1}^n \delta_i} e^{-\lambda \sum_{i=1}^n T_i}
\]
By applying Bayes' Theorem, the posterior distribution is proportional to the product of the likelihood and the prior:
\[
p(\lambda \mid T_1, \dots, T_n, \delta_1, \dots, \delta_n) \propto L(\lambda) \pi(\lambda)
\]
Substituting the likelihood and prior expressions:
\[
p(\lambda \mid T_1, \dots, T_n, \delta_1, \dots, \delta_n) \propto \lambda^{\alpha + \sum_{i=1}^n \delta_i - 1} e^{-\lambda \left( \beta + \sum_{i=1}^n T_i \right)}
\]
This is a Gamma distribution with updated parameters \(\alpha' = \alpha + \sum_{i=1}^n \delta_i\) and \(\beta' = \beta + \sum_{i=1}^n T_i\), thus:
\[
\lambda \mid \left(T_i, \delta_i\right)_{i = 1}^{n}  \sim \text{Gamma}\left( \alpha + \sum_{i=1}^n \delta_i, \, \beta + \sum_{i=1}^n T_i \right)
\]
\end{proof}

\begin{proof}[Proof of Theorem~\ref{thm: 3}]
The survival function \(S(t)\) for the exponential distribution is:
\[
S(t) = P(T > t) = e^{-\lambda t}
\]
Setting \(S(t_{\text{med}}) = 0.5\), we solve:
\[
e^{-\lambda t_{\text{med}}} = 0.5 \quad \Rightarrow \quad t_{\text{med}} = \frac{\ln 0.5}{-\lambda} = \frac{\ln 2}{\lambda}
\]
\end{proof}
\end{document}